\newif\ifEntropy
\Entropytrue
\Entropyfalse

\ifEntropy
\documentclass[entropy,article,submit,oneauthor,pdftex,10pt,a4paper]{mdpi}

\firstpage{1} 
\articlenumber{x}
\doinum{10.3390/------}
\pubvolume{xx}
\pubyear{2017}
\copyrightyear{2017}
\externaleditor{Academic Editor: name}
\history{Received: date; Accepted: date; Published: date}

\pdfoutput=1

\renewcommand{\ul}[1]{\underline{#1}}
\else
\documentclass[a4paper]{amsart}

\pdfoutput=1

\newtheorem{Lemma}{Lemma}\newtheorem{Proposition}{Proposition}\newtheorem{Theorem}{Theorem}
\theoremstyle{definition}\newtheorem{Definition}{Definition}
\theoremstyle{remark}\newtheorem{Example}{Example}\newtheorem{Remark}{Remark}
\usepackage{microtype}
\usepackage[utf8]{inputenc}
\usepackage{natbib}
\usepackage{hyperref}

\newcommand{\Title}{\title}
\newcommand{\Author}{\author}
\newcommand{\keyword}{\keywords}
\newcommand{\MSC}{\subjclass[2000]}

\newcommand{\ul}[1]{\underline{#1}}
\fi

\usepackage{tikz}
\usetikzlibrary{arrows,shapes}
\usepackage{siunitx}

\newcommand{\Acal}{\mathcal{A}}

\newcommand{\PIdrei}[9]{
    \node (X123) at (3,6)        {#1};
    \node (X12) at (1.5,5)       {#2};
    \node (X13) at (3  ,5)       {#3};
    \node (X23) at (4.5,5)       {#4};
    \node (X12X13) at (1.5,4)    {#5};
    \node (X12X23) at (3  ,4)    {#6};
    \node (X13X23) at (4.5,4)    {#7};
    \node (X1) at (0.5,3)        {#8};
    \node (X2) at (2. ,3)        {#9};
    \PIdreiTeilzwei
}
\newcommand{\PIdreiTeilzwei}[9]{
    \node (X3) at (3.5,3)        {#1};
    \node (X12X13X23) at (6 ,3)  {#2};
    \node (X1X23) at (1.5,2)     {#3};
    \node (X2X13) at (3  ,2)     {#4};
    \node (X3X12) at (4.5,2)     {#5};
    \node (X1X2) at (1.5,1)      {#6};
    \node (X1X3) at (3  ,1)      {#7};
    \node (X2X3) at (4.5,1)      {#8};
    \node (X1X2X3) at (3,0)      {#9};
    \begin{scope}[->,>=stealth]
      \path (X123) edge (X12);  \path (X123) edge (X13);  \path (X123) edge (X23);
      \path (X12) edge (X12X13);  \path (X12) edge (X12X23);
      \path (X13) edge (X12X13);  \path (X13) edge (X13X23);
      \path (X23) edge (X12X23);  \path (X23) edge (X13X23);
      \path (X12X13) edge (X1);  \path (X12X13) edge (X12X13X23);
      \path (X12X23) edge (X2);  \path (X12X23) edge (X12X13X23);
      \path (X13X23) edge (X3);  \path (X13X23) edge (X12X13X23);
      \path (X1) edge (X1X23);  \path (X12X13X23) edge (X1X23);
      \path (X2) edge (X2X13);  \path (X12X13X23) edge (X2X13);
      \path (X3) edge (X3X12);  \path (X12X13X23) edge (X3X12);
      \path (X1X23) edge (X1X2);  \path (X1X23) edge (X1X3);
      \path (X2X13) edge (X1X2);  \path (X2X13) edge (X2X3);
      \path (X3X12) edge (X1X3);  \path (X3X12) edge (X2X3);
      \path (X1X2) edge (X1X2X3);  \path (X1X3) edge (X1X2X3);  \path (X2X3) edge (X1X2X3);
    \end{scope}
}

\ifEntropy\else
\begin{document}
\fi

\Title{Secret Sharing and Shared Information}
\Author{Johannes Rauh}
\ifEntropy
\AuthorNames{Johannes Rauh}
\else
\email{jrauh@mis.mpg.de}
\fi

\address{%
%$^{1}$ ; e-mail@e-mail.com\\
Max Planck Institute for Mathematics in the Sciences, Leipzig, Germany\ifEntropy{} ; jrauh@mis.mpg.de\fi}

\newcommand{\abstracttext}{ Secret sharing is a cryptographic discipline in which the goal is to distribute information about a secret
  over a set of participants in such a way that only specific authorized combinations of participants together can
  reconstruct the secret.  Thus, secret sharing schemes are systems of variables in which it is very clearly specified
  which subsets have information about the secret.  As such, they provide perfect model systems for information
  decompositions.  However, following this intuition too far leads to an information decomposition with negative partial
  information terms, which are difficult to interpret.  One possible explanation is that the partial information lattice
  proposed by Williams and Beer is incomplete and has to be extended to incorporate terms corresponding to higher order
  redundancy.  These results put bounds on information decompositions that follow the partial information framework, and
  they hint at where the partial information lattice needs to be improved.}
\ifEntropy
\abstract{\abstracttext}
\else
\begin{abstract}
\abstracttext
\end{abstract}
\fi

\keyword{Information decomposition, partial information lattice, shared information, secret sharing}

\MSC{94A17; %  	Measures of information, entropy
  94A62} %   	Authentication and secret sharing

\ifEntropy
\begin{document}
\else
\maketitle
\fi

\section{Introduction}

\citet{WilliamsBeer:Nonneg_Decomposition_of_Multiinformation} have proposed a general framework to decompose the
multivariate mutual information $I(S;X_{1},\dots,X_{n})$ between a target random variable~$S$ and predictor random
variables $X_{1},\dots,X_{n}$ into different terms (called \emph{partial information terms}) according to different ways
in which combinations of the variables $X_{1},\dots,X_{n}$ provide unique, shared or synergistic information about~$S$.
Williams and Beer argue that such a decomposition can be based on a measure of shared information.  The underlying idea
is that any information can be classified according to ``who knows what.''  But is this true?

A situation where the question ``who knows what'' is easy to answer very precisely is secret sharing, a part of
cryptography in which the goal is to distribute information (the \emph{secret}) over a set of participants such that the
secret can only be reconstructed if certain \emph{authorized} combinations of the participants join their information
(see~\cite{Beimel:Secret_sharing_survey} for a survey).  The set of authorized combinations is called the \emph{access
  structure}.  Formally, the secret is modelled as a random variable~$S$, and a \emph{secret sharing scheme} assigns a
random variable $X_{i}$ to each participant~$i$ in such a way that, if $\{i_{1},\dots,i_{k}\}$ is an authorized set of
participants, then $S$ is a function of $X_{i_{1}},\dots,X_{i_{k}}$; that is, $H(S|X_{i_{1}},\dots,X_{i_{k}}) = 0$; and,
conversely, if $\{i_{1},\dots,i_{k}\}$ is not authorized, then $H(S|X_{i_{1}},\dots,X_{i_{k}}) > 0$.  It is assumed that
the participants know the scheme, and so any authorized combination of participants can reconstruct the secret if they
join their information.  A secret sharing scheme is \emph{perfect} if non-authorized sets of participants know nothing
about the secret; i.e., $H(S|X_{i_{1}},\dots,X_{i_{k}}) = H(S)$.  Thus, in a perfect secret sharing scheme, it is very
clearly specified ``who knows what.''  In this sense, perfect secret sharing schemes provide model systems for which it
should be easy to write down an information decomposition.

One connection between secret sharing and information decompositions is that the set of access structures of secret
sharing schemes with $n$ participants is in one-to-one correspondence with the partial information terms of Williams and
Beer.
% Both are determined by antichains, i.e.~families of subsets $A_{1},\dots,A_{k}\subseteq\{1,\dots,n\}$ with
% $A_{i}\not\subseteq A_{j}$ for $i\neq j$.  Partial information terms correspond to antichains, because if
% $A_{i}\subseteq A_{j}$, then $A_{j}$ knows everything that $A_{i}$ knows, and so the shared information of $A_{i}$
% and~$A_{j}$ is always part the information of~$A_{i}$.
This correspondence makes it possible to give another interpretation to all partial information terms: Namely, the
partial information term is a measure of how similar a given system of random variables is to a secret sharing scheme
with a given access structure.

This correspondence also allows to introduce the \emph{secret sharing property} that makes precise the above intuition:
An information decomposition satisfies this property if and only if any perfect secret sharing scheme has just a single
partial information term (which corresponds to its access structure).  Lemma~\ref{lem:WB-implies-ssp} states the secret
sharing property is implied by the Williams and Beer axioms, which shows that the secret sharing property plays well
together with the ideas of Williams and Beer.  Proposition~\ref{prop:prescribing} shows that in an information
decomposition that satisfies a natural generalization of this property, it is possible to prescribe arbitrary
nonnegative values to all partial information terms.

These results suggest that perfect secret sharing schemes fit well together with the ideas of Williams and Beer.
However, following this intuition too far leads to inconsistencies.  As Theorem~\ref{sec:incomp-result} shows, extending
the secret sharing property to pairs of perfect secret sharing schemes leads to negative partial information terms.
While other authors have started to build an intuition for negative partial terms and argue that they may be unavoidable
in information decompositions, the concluding section collects arguments against such claims and proposes as another
possible solutions that the Williams and Beer framework is incomplete and is missing nodes that represent higher
order redundancy.

\medskip%
Cryptography, where the goal is not only to transport information (as in coding theory) but also to keep it concealed
from unauthorized parties, has initiated many interesting developments in information theory, for example, by
introducing new information measures and re-interpreting older ones; see, for example,
\cite{MaurerWolf97:intrinsic_conditional_MI,CsiszarNarayan04:SEcrecy_capacities_for_multiple_terminals}.  This
manuscript focuses on another contribution of cryptography: probabilistic systems with well-defined distribution of
information.

\medskip%
The remainder of this article is organized as follows: Section~\ref{sec:psss} summarizes definitions and results about
secret sharing schemes.  Section~\ref{sec:info-dcomp-psss} introduces different secret sharing properties that fix the values
that a measure of shared information assigns to perfect secret sharing schemes and combinations thereof.  The main
result of Section~\ref{sec:incomp-result} is that the pairwise secret sharing property leads to negative partial
information terms.  Section~\ref{sec:discussion} discusses the implications of this incompatibility result.

%%%%%%%%%%%%%%%%%%%%%%%%%%%%%%%%%%%%%%%%%%
\section{Perfect secret sharing schemes}
\label{sec:psss}

We consider $n$ participants among whom we want to distribute information about a secret in such a way that we can
control which subsets of participants together can decrypt the secret.
\begin{Definition}
  An \emph{access structure} $\Acal$ is a family of subsets of $\{1,\dots,n\}$, closed to taking supersets.
  Elements of $\Acal$ are called \emph{authorized sets}.

  A \emph{secret sharing scheme} with access structure $\Acal$ is a family of random variables $S$, $X_{1},\dots,X_{n}$
  such that:
  \begin{itemize}
  \item $H(X_{A},S) = H(X_{A})$, whenever $A\in\Acal$.
  \end{itemize}
  Here, $X_{A}=(X_{i})_{i\in A}$ for all subsets $A\subseteq\{1,\dots,n\}$.
  A secret sharing scheme is \emph{perfect} if
  \begin{itemize}
  \item $H(X_{A},S) = H(X_{A}) + H(S)$, whenever $A\notin\Acal$.
  \end{itemize}
\end{Definition}
The condition for perfection is equivalent to~$H(S|X_{A}) = H(S)$.
See~\cite{Beimel:Secret_sharing_survey} for a survey on secret sharing.

\begin{Theorem}
  \label{thm:psss-existence}
  For any access structure~$\Acal$ and any $h>0$, there exists a perfect secret sharing scheme with access
  structure~$\Acal$ for which the entropy of the secret $S$ equals~$H(S) = h$.
\end{Theorem}
\begin{proof}
  Perfect secret sharing schemes for arbitrary access structures were first constructed
  by~\citet{ItoSaitoNishizeki87:General_secret_sharing_schemes}.  In this construction, the entropy of the secret equals
  \SI{1}{bit}.  Combining $n$ copies of such a secret sharing scheme gives a secret sharing scheme with a secret of $n$~bit.
  As explained in~\cite[Claim~1]{Beimel:Secret_sharing_survey}, the distribution of the secret may be perturbed
  arbitrarily (as long as the support of the distribution remains the same).  In this way it is possible to prescribe
  the entropy of the secret in a perfect secret sharing scheme.
\end{proof}

\begin{Example}
  \label{ex:psss}
  Let $Y_{1},Y_{2},Y_{3},S$ be independent uniform binary random variables, and let $A = (Y_{1}, Y_{2}\oplus S)$, $B =
  (Y_{2}, Y_{3}\oplus S)$, $C = (Y_{3}, Y_{1}\oplus S)$, where $\oplus$ denotes addition modulo 2 (or the
  $\operatorname{XOR}$ operation).  Then $(S, A, B, C)$ is a perfect secret sharing scheme with access structure
  \begin{equation*}
    \{A,B\}, \{A,C\}, \{B,C\}, \quad\{A,B,C\}.
  \end{equation*}
\end{Example}

It may be of little surprise that integer addition modulo $k$ is an important building block in many secret sharing
schemes.

While existence of perfect secret sharing schemes is solved, there remains the problem of finding efficient secret
sharing schemes in the sense that the variables $X_{1},\dots,X_{n}$ should be as small as possible (in the sense of a
small entropy), given a fixed entropy of the secret.  For instance, in Example~\ref{ex:psss}, $H(X_{i})/H(S) = 2$ for
all~$i$.  See~\cite{Beimel:Secret_sharing_survey} for a survey.
% To make this precise, the following definitions are needed: A variable $X_{i}$ is \emph{essential} for $\Acal$ if
% there exists $A\in\Acal$ with $A\setminus\{i\}\notin\Acal$.  Clearly, if $X_{i}$ is essential in a perfect secret
% sharing scheme, then $H(X_{i})\ge H(S)$.  Thus, for a good secret sharing scheme, $\max_{i}\frac{H(X_{i})}{H(S)}$ will
% be close to~1.

\bigskip

Since an access structure~$\Acal$ is closed to taking supersets, it is uniquely determined by its inclusion-minimal
elements
\begin{equation*}
  \ul\Acal := \big\{ A \in\Acal : \text{if }B\subseteq A\text{ and } B\neq A, \text{ then }B\notin\Acal \big\}.
\end{equation*}
For instance, in Example~\ref{ex:psss}, the first three elements belong to~$\ul\Acal$.  The set $\ul\Acal$ has the
property that no element of $\ul\Acal$ is a subset of another element of~$\ul\Acal$.  Such a collection of sets is
called an \emph{antichain}.  Conversely, any such antichain equals the set of inclusion-minimal elements of a unique
access structure.

The antichains have a natural lattice structure, which was used by Williams and Beer to order the different values of
shared information and organize them into what they call the \emph{partial information lattice}.  The same lattice also
has a description in terms of secret sharing.
\begin{Definition}
  \label{def:PI-lattice}
  Let $(A_{1},\dots,A_{k})$ and $(B_{1},\dots,B_{l})$ be antichains.  Then
  \begin{equation*}
    (A_{1},\dots,A_{k}) \preceq (B_{1},\dots,B_{l})
    \quad:\Longleftrightarrow\quad
    \text{for any }B_{i}\text{ there exists }A_{j}\text{ with }A_{j}\subseteq B_{i}.
  \end{equation*}
\end{Definition}

The partial information lattice for the case $n=3$ is depicted in Figure~\ref{fig:xor3}.

\begin{Lemma}
  \label{lem:lattice-authorized}
  Let $\Acal$ be an access structure on~$\{1,\dots,n\}$, and let $(B_{1},\dots,B_{l})$ be an antichain.  Then
  $B_{1},\dots,B_{l}$ are all authorized for~$\Acal$ if and only if $\ul\Acal\preceq(B_{1},\dots,B_{l})$.
\end{Lemma}
\begin{proof}
  The statement directly follows from the definitions.
\end{proof}

%%%%%%%%%%%%%%%%%%%%%%%%%%%%%%%%%%%%%%%%%%
\section{Information decompositions of secret sharing schemes}
\label{sec:info-dcomp-psss}

\citet{WilliamsBeer:Nonneg_Decomposition_of_Multiinformation} proposed to decompose the total mutual information
$I(S;X_{1},\dots,X_{n})$ between a target random variable~$S$ and predictor random variables $X_{1},\dots,X_{n}$
according to different ways in which combinations of the variables $X_{1},\dots,X_{n}$ provide unique, shared or
synergistic information about~$S$.  One of their main ideas is to base such a decomposition on a single \emph{measure of
  shared information}~$I_{\cap}$, which is a function $I(S;Y_{1},\dots,Y_{k})$ that takes as arguments a list of random
variables, of which the first, $S$, takes a special role.  To arrive at a decomposition of $I(S;X_{1},\dots,X_{n})$, the
variables $Y_{1},\dots,Y_{k}$ are taken to be combinations $X_{A} = (X_{i})_{i\in A}$ of $X_{1},\dots,X_{n}$,
corresponding to subsets $A$ of $\{1,\dots,n\}$.  For simplicity, $I_{\cap}(S;X_{A_{1}},\dots,X_{A_{k}})$ is denoted by
$I_{\cap}(S;A_{1},\dots,A_{k})$ for all $A_{1},\dots,A_{k}\subseteq\{1,\dots,n\}$.

Williams and Beer proposed a list of axioms that such a measure $I_{\cap}$ should satisfy.  It follows from these axioms
that it suffices to consider the function $I_{\cap}(S; A_{1},\dots,A_{k})$ in the case that $(A_{1},\dots,A_{k})$ is an
antichain.  Moreover, $I_{\cap}(S;\;\cdot)$ is a monotone function on the partial information lattice
(Definition~\ref{def:PI-lattice}).
Thus it is natural to write each value $I_{\cap}(S; A_{1},\dots,A_{k})$ on the lattice as a sum of local terms
$I_{\partial}$ corresponding to the antichains that lie below $(A_{1},\dots,A_{k})$ in the lattice:
\begin{equation*}
  I_{\cap}(S; A_{1},\dots,A_{k}) = \sum_{(B_{1},\dots,B_{l})\preceq(A_{1},\dots,A_{k})}I_{\partial}(S; B_{1},\dots,B_{l}).
\end{equation*}
The terms $I_{\partial}$ are called \emph{partial information terms}.  This representation always exists, and the
partial information terms are uniquely defined (using a Möbius inversion).  However, it is not guaranteed that
$I_{\partial}$ is always nonnegative.  If $I_{\partial}$ is nonnegative, then $I_{\cap}$ is called \emph{locally
  positive}.

Williams and Beer also defined a function denoted by $I_{\min}$ that satisfies their axioms and that is locally
positive.  While the framework is intriguing and has attracted a lot of further research (as this special issue
illustrates), the function $I_{\min}$ has been critiziced as not measuring the right thing.  The difficulty of finding a
reasonable measure of shared information that is locally
positive~\citep{BROJ13:Shared_information,RBOJ14:Reconsidering_unique_information} has led some to argue that maybe
local positivity is not a necessary requirement for an information decomposition.  This issue is discussed further in
Section~\ref{sec:discussion}.

The goal of this section is to present additional natural properties for a measure of shared information that relate
secret sharing with the intuition behind information decompositions.
In a perfect secret sharing scheme, any combination of participants knows either nothing or everything about~$S$.  This
motivates the following definition:

\begin{Definition}
  A measure of shared information $I_{\cap}$ has the \emph{secret sharing property} if and only if for any access
  structure $\Acal$ and any perfect secret sharing scheme $(X_{1},\dots,X_{n},S)$ with access structure~$\Acal$, the
  following holds:
  \begin{equation*}
    I_{\cap}(S;A_{1},\dots,A_{k}) =
    \begin{cases}
      H(S), & \text{ if }A_{1},\dots,A_{k}\text{ are all authorized}, \\
      0, & \text{ otherwise,}
    \end{cases}
\ifEntropy
    \qquad
    \text{for all $A_{1},\dots,A_{k}\subseteq\{1,\dots,n\}$.}
\fi
  \end{equation*}
\ifEntropy\else
  for any $A_{1},\dots,A_{k}\subseteq\{X_{1},\dots,X_{n}\}$.
\fi
\end{Definition}
\begin{Lemma}
  \label{lem:WB-implies-ssp}
  The secret sharing property is implied by the Williams and Beer axioms.
\end{Lemma}
\begin{proof}
  The Williams and Beer axioms imply that
  \begin{equation*}
    I_{\cap}(S;A_{1},\dots,A_{k}) \le I(S;A_{i}) = 0
  \end{equation*}
  whenever $A_{i}$ is not authorized.  On the other hand, when $A_{1},\dots,A_{k}$ are all authorized, then the
  monotonicity axiom implies
  \begin{equation*}
    I_{\cap}(S;A_{1},\dots,A_{k}) \ge
    I_{\cap}(S;A_{1},\dots,A_{k},S) = I_{\cap}(S;S) = H(S).
    \qedhere
  \end{equation*}
\end{proof}
Perfect secret sharing schemes lead to information decompositions with a single nonzero partial information term:
\begin{Lemma}
  \label{lem:local-psss-terms}
  If $I_{\cap}$ has the secret sharing property and if $(X_{1},\dots,X_{n},S)$ is a perfect secret sharing scheme with
  access structure~$\Acal$, then
  \begin{equation}
    \label{eq:Ipartial}
    I_{\partial}(S;A_{1},\dots,A_{k}) = 
    \begin{cases}
      H(S), & \text{ if }\ul{\Acal} = \{A_{1},\dots,A_{k}\}, \\
      0, & \text{ otherwise,}
    \end{cases}
\ifEntropy
    \qquad
    \text{for all $A_{1},\dots,A_{k}\subseteq\{1,\dots,n\}$.}
\fi
  \end{equation}
\ifEntropy\else
  for any $A_{1},\dots,A_{k}\subseteq\{X_{1},\dots,X_{n}\}$.
\fi
\end{Lemma}
\begin{proof}
  Suppose that $\ul{\Acal} = \{A'_{1},\dots,A'_{k'}\}$, and let $J_{\partial}(S;A_{1},\dots,A_{k})$ be the right hand
  side of~\eqref{eq:Ipartial}.  We need to show that $I_{\partial} = J_{\partial}$.  Since the Möbius inversion is
  unique, it suffices to show that $J_{\cap} = I_{\cap}$, where
  \begin{equation*}
    J_{\cap}(S;A_{1},\dots,A_{k})
    = \sum_{(B_{1},\dots,B_{l})\preceq(A_{1},\dots,A_{k})} J_{\partial}(S;B_{1},\dots,B_{l}).
  \end{equation*}
  By Lemma~\ref{lem:lattice-authorized},
  \begin{equation*}
    J_{\cap}(S;A_{1},\dots,A_{k}) =
    \begin{cases}
      H(S), & \text{ if }A_{1},\dots,A_{k}\text{ are all authorized}, \\
      0, & \text{ otherwise,}
    \end{cases}
  \end{equation*}
  for any $A_{1},\dots,A_{k}\subseteq\{X_{1},\dots,X_{n}\}$, from which the claim follows.
\end{proof}

What happens when we have several secret sharing schemes involving the same participants?  In order to have a clear
intuition, assume that the secret sharing schemes satisfy the following definition:
\begin{Definition}
  \label{def:sss_combi}
  Let $\Acal_{1},\dots,\Acal_{l}$ be access structures on $\{1,\dots,n\}$.  A \emph{combination of (perfect) secret
    sharing schemes} with access structures~$\Acal_{1},\dots,\Acal_{l}$ consists of random variables $S_{1}, \dots,
  S_{l}$, $X_{1},\dots,X_{n}$ such that $(S_{i},X_{1},\dots,X_{n})$ is a (perfect) secret sharing scheme with access
  structure $\Acal_{i}$ for $i=1,\dots,l$ and such that
  \begin{equation*}
    H(S_{i}|S_{1},\dots,S_{i-1},S_{i+1},\dots,S_{l},X_{A}) =
    H(S_{i}) \text{ if } A\notin\Acal_{i}. 
  \end{equation*}
\end{Definition}
This definition ensures that the secrets are independent in the sense that knowing some of the secrets provides no
information about the other secrets.  Formally, one can see that the secrets are probabilistically independent as
follows: For any $A\notin\Acal_{i}$ (for example, $A=\emptyset$),
\begin{equation*}
  H(S_{i}|S_{1},\dots,S_{i-1},S_{i+1},\dots,S_{l}) \ge
  H(S_{i}|S_{1},\dots,S_{i-1},S_{i+1},\dots,S_{l},X_{A}) =
  H(S_{i}).
\end{equation*}

In Definition~\ref{def:sss_combi}, if two access structures $\Acal_{i},\Acal_{j}$ are identical, then we can replace
$S_{i}$ and $S_{j}$ by a single random variable $(S_{i},S_{j})$ and obtain a smaller combination of (perfect) secret
sharing schemes.

In a combination of perfect secret sharing schemes, it is very clear who knows what: Namely, a group of participants
knows all secrets for which it is authorized, while it knows nothing about the remaining secrets.  This motivates the
following definition:
\begin{Definition}
  \label{def:comb-ssp}
  A measure of shared information $I_{\cap}$ has the \emph{combined secret sharing property} if and only if for any
  combination of perfect secret sharing schemes with access structures $\Acal_{1},\dots,\Acal_{l}$,
  \begin{equation}
    \label{eq:comb-ssp}
    I_{\cap}\big((S_{1},\dots,S_{l});A_{1},\dots,A_{k}\big) =
    H\big(\{ S_{i} : A_{1},\dots,A_{k}\in\Acal_{i}\}\big)
  \end{equation}
  (the entropy of those secrets for which $A_{1},\dots,A_{k}$ are all authorized).
  $I_{\cap}$ has the \emph{pairwise secret sharing property} if and only if the same holds true in the special
  case~$l=2$.
\end{Definition}
The combined secret sharing property implies the pairwise secret sharing property.  The pairwise secret sharing property
does not follow from the Williams and Beer axioms.  For example, % Williams and Beer's measure
$I_{\min}$ satisfies the Williams and Beer axioms, but not the pairwise secret sharing property (as will become apparent
in Theorem~\ref{thm:incompatibility}).
So one can ask whether the pairwise and combined secret sharing properties are compatible with the Williams and Beer
axioms.  This question is difficult to answer, since currently there are only two proposed measures of shared
information that satisfy the Williams and Beer axioms, namely $I_{\min}$ and the \emph{minimum of mutual informations}
\citep{Barrett2014:Gaussian_information_decomposition}
\begin{equation*}
  I_{\text{MMI}}(S;A_{1},\dots,A_{k}) := \min_{i=1,\dots,k} I(S;A_{i}).
\end{equation*}
Both measures do not satisfy the pairwise secret sharing property.

While there has been no further proposal for a function that satisfies the Williams and Beer axioms for arbitrarily many
arguments, several measures have been proposed for the ``bivariate case'' $k=2$, notably $I_{\text{red}}$ of
\citet{HarderSalgePolani2013:Bivariate_redundancy} and $\widetilde{SI}$ of
\cite{BROJA13:Quantifying_unique_information}.  The appendix shows that $\widetilde{SI}$ at least satisfies the combined
secret sharing property ``as far as possible.''

\medskip%
Combinations of $l$ perfect secret sharing schemes lead to information decompositions with at most $l$ nonzero partial
information terms.
\begin{Lemma}
  \label{lem:partial-perfect-combination}
  Assume that $I_{\cap}$ has the combined secret sharing property.  If $(S_{1},\dots,S_{l}$, $X_{1},\dots,X_{n})$ is a
  combination of perfect secret sharing schemes with pairwise different
  access structures~$\Acal_{1},\dots,\Acal_{l}$, then
  \begin{equation*}
%    \label{eq:Ipartial-combi}
    I_{\partial}\big((S_{1},\dots,S_{l});A_{1},\dots,A_{k}\big) =
    \begin{cases}
      H(S_{i}), & \text{ if }\ul{\Acal_{i}} = \{A_{1},\dots,A_{k}\}
      \ifEntropy\else\\&\qquad\qquad\fi \text{ for some }i\in\{1,\dots,l\}, \\
      0, & \text{ otherwise,}
    \end{cases}
  \end{equation*}
  for any $A_{1},\dots,A_{k}\subseteq\{X_{1},\dots,X_{n}\}$.
\end{Lemma}
The proof is similar to the proof of Lemma~\ref{lem:local-psss-terms} and omitted.

The combined secret sharing property implies that any combination of nonnegative values can be prescribed as partial
information values.
\begin{Proposition}
  \label{prop:prescribing}
  Suppose that a nonnegative number $h_{\Acal}$ is given for any antichain~$\Acal$.  For any measure of shared
  information that satisfies the combined secret sharing property, there exist random variables~$S,X_{1},\dots,X_{n}$
  such that the corresponding partial measure $I_{\partial}$ satisfies $I_{\partial}(S;A_{1},\dots,A_{k}) =
  h_{A_{1},\dots,A_{k}}$ for all antichains $\Acal=(A_{1},\dots,A_{k})$.
\end{Proposition}
\begin{proof}
  By Theorem~\ref{thm:psss-existence}, for each antichain $\Acal$ there exists a perfect secret sharing scheme
  $S_{\Acal}$, $X_{1,\Acal},\dots,X_{n,\Acal}$ with $H(S_{\Acal}) = h_{\Acal}$.  Combine independent copies of these
  perfect secret sharing schemes and let
  \begin{equation*}
    S = (S_{\Acal})_{\Acal}, \quad X_{1} = (X_{1,\Acal})_{\Acal}, \quad \dots, \quad X_{n} = (X_{n,\Acal})_{\Acal},
  \end{equation*}
  where $\Acal$ runs over all antichains.  Then $S,X_{1},\dots,X_{n}$ is an independent combination of perfect secret
  sharing schemes, and the statement follows from Lemma~\ref{lem:partial-perfect-combination}.
\end{proof}

Unfortunately, not every random variable $S$ can be decomposed in such a way as a combination of secret sharing schemes.
However, Proposition~\ref{prop:prescribing} suggests that, given a measure $I_{\cap}$ of shared information that
satisfies the combined secret sharing property, $I_{\partial}(S;\ul\Acal)$ can informally be interpreted as a
measure that quantifies how much $(X_{1},\dots,X_{n},S)$ looks like a perfect secret sharing scheme with access
structure~$\Acal$.

\begin{Lemma}
  \label{lem:weak-identity}
  Suppose that $I_{\cap}$ is a measure of shared information that satisfies the pairwise secret sharing property.
  If $X_{1}$ and $X_{2}$ are independent, then \ifEntropy$\else$$\fi I_{\cap}\big((X_{1},X_{2}); X_{1}, X_{2}\big) = 0\ifEntropy$.\else.$$\fi
\end{Lemma}
In the language of~\cite{Ince17:Iccs}, the lemma says that the pairwise secret sharing property implies the
\emph{independent identity property}.
\begin{proof}
  Let $S_{1} = X_{1}$, $S_{2} = X_{2}$.  Then $S_{1}, S_{2}, X_{1}, X_{2}$ is a pair of perfect secret sharing schemes
  with access structures $\Acal_{1}=\big\{\{1\}\big\}$ and $\Acal_{2}=\big\{\{2\}\big\}$.  The statement follows from
  Definition~\ref{def:comb-ssp}, since $X_{1}$ is not authorized for~$\Acal_{2}$ and $X_{2}$ is not authorized
  for~$\Acal_{1}$.
\end{proof}

\section{Incompatibility with local positivity}
\label{sec:incomp-result}

Unfortunately, although the combined secret sharing property very much fits the intuition behind the axioms of Williams
and Beer, it is incompatible with a nonnegative decomposition according to the partial information lattice:
\begin{Theorem}
  \label{thm:incompatibility}
  Let $I_{\cap}$ be a measure of shared information that satisfies the Williams-Beer axioms and has the pairwise secret
  sharing property.  Then $I_{\partial}$ is not nonnegative.
\end{Theorem}

\begin{proof}
  The XOR example, which was already used by~\citet{BROJ13:Shared_information} and
  \citet{RBOJ14:Reconsidering_unique_information} to prove incompatibility results for properties of information
  decompositions, can also be used here.

  Let $X_{1},X_{2}$ be independent binary uniform random variables, let $X_{3}=X_{1}\oplus X_{2}$, and let $S =
  (X_{1},X_{2},X_{3})$.  Observe that the situation is symmetric in $X_{1},X_{2},X_{3}$.  In particular,
    $X_{2},X_{3}$ are also independent, and $X_{1}=X_{2}\oplus X_{3}$.  The following values of $I_{\cap}$ can be
  computed from the assumptions:
  \begin{itemize}
  \item $I_{\cap}\big(S;X_{1},(X_{2}X_{3})\big) = I_{\cap}\big(S; X_{1},(X_{1}X_{2}X_{3})\big) = I_{\cap}(S;X_{1}) = \SI{1}{bit}$, since $X_{1}$ is a function of
    $(X_{2},X_{3})$ and by the monotonicity axiom.
  \item $I_{\cap}(S;X_{1},X_{2}) = I_{\cap}\big((X_{1}X_{2}X_{3});X_{1},X_{2}\big) = I_{\cap}\big((X_{1}X_{2});X_{1},X_{2}\big) = 0$ by
    Lemma~\ref{lem:weak-identity}.
  \end{itemize}
  By monotonicity, $I_{\cap}(S;X_{1},X_{2},X_{3}) = 0$.  Moreover,
  \ifEntropy$\else$$\fi I_{\cap}\big(S;(X_{1}X_{2}),(X_{1}X_{3}),(X_{2}X_{3})\big)\le\SI{2}{bit}\ifEntropy$,\else,$$\fi{} since \SI{2}{bit} is the total entropy in the
  system.  But then
  \begin{multline*}
    I_{\partial}\big(S;(X_{1}X_{2}),(X_{1}X_{3}),(X_{2}X_{3})\big)
    = I_{\cap}\big(S;(X_{1}X_{2}),(X_{1}X_{3}),(X_{2}X_{3})\big)\\
    \shoveright{
    - I_{\cap}\big(S;X_{1},(X_{2}X_{3})\big) - I_{\cap}\big(S;X_{2},(X_{1}X_{3})\big) - I_{\cap}\big(S;X_{3},(X_{1}X_{2})\big)
    \pm 0
    }
    \\
    \le \SI{2}{bit} - \SI{3}{bit} = \SI{-1}{bit},
  \end{multline*}
  where $\pm0$ denotes values of $I_{\cap}$ that vanish.  Thus, $I_{\partial}$ is not nonnegative.
\end{proof}
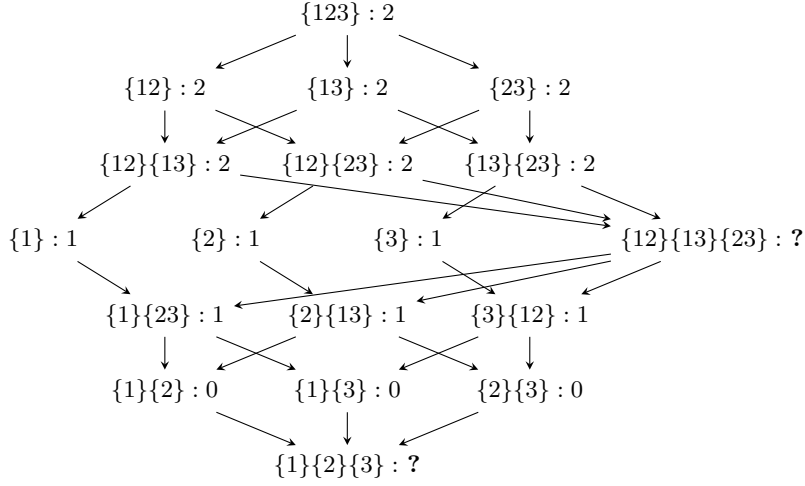
\begin{figure}
  \begin{center}
    \begin{tikzpicture}[xscale=1.6, yscale=1] %
      \PIdrei{\small $\{123\}:2$}
      {\small $\{12\}:2$}
      {\small $\{13\}:2$}
      {\small $\{23\}:2$}
      {\small $\{12\}\{13\}:2$}
      {\small $\{12\}\{23\}:2$}
      {\small $\{13\}\{23\}:2$}
      {\small $\{1\}:1$}
      {\small $\{2\}:1$}
      {\small $\{3\}:1$}
      {\small $\{12\}\{13\}\{23\}:$ \textbf{?}}
      {\small $\{1\}\{23\}:1$}
      {\small $\{2\}\{13\}:1$}
      {\small $\{3\}\{12\}:1$}
      {\small $\{1\}\{2\}:0$}
      {\small $\{1\}\{3\}:0$}
      {\small $\{2\}\{3\}:0$}
      {\small $\{1\}\{2\}\{3\}:$ \textbf{?}}
    \end{tikzpicture}
  \end{center}
  \caption{ The partial information lattice for~$n=3$.  Each node is indexed by an antichain.
    % , where, as usual, $X_{i}$ is abbreviated by its index~$i$.
    The values (in bit) of the shared information in the XOR example from the proof of Theorem~\ref{thm:incompatibility}
    according to the pairwise secret sharing property are given after the colon. }
  \label{fig:xor3}
\end{figure}
Note that the random variables $(S=(X_{1},X_{2},X_{3}),X_{1},X_{2},X_{3})$ from the proof of
Theorem~\ref{thm:incompatibility} form three perfect secret sharing schemes that do not satisfy the definition of a
combination of perfect secret sharing schemes.  The three secrets $X_{1},X_{2},X_{3}$ are not independent, but they are
pair-wise independent (and so Lemma~\ref{lem:partial-perfect-combination} does not apply).

\begin{Remark}
  The XOR example from the proof of Theorem~\ref{thm:incompatibility} (which was already used by
  \citet{BROJ13:Shared_information} and \citet{RBOJ14:Reconsidering_unique_information}) was criticized by
  \citet{ChicharroPanzeri17:Dual_Decompositions} on the grounds that it involves random variables that stand in a
  deterministic functional relation (in the sense that~$X_{3} = X_{1}\oplus X_{2}$).  Chicharro and Panzeri argue that
  in such a case it is not appropriate to use the full partial information lattice.  Instead, the functional
  relationship should be used to eliminate (or identify) nodes from the lattice.  Thus, while the monotonicity axiom of
  Williams and Beer implies $I_{\cap}(S;X_{3},(X_{2},X_{3})) = I_{\cap}(S;X_{3})$ (and so $\{3;23\}$ is not part of the
  partial information lattice), the same axiom also implies that $I_{\cap}(S;X_{3},(X_{1},X_{2})) = I_{\cap}(S;X_{3})$
  in the XOR example, and so $\{3;12\}$ should similarly be excluded from the lattice when analyzing this particular
  example.  But note that the first argument is a formal argument that is valid for all joint distributions of
  $S,X_{1},X_{2},X_{3}$, while the second argument takes into account the particular underlying distribution.

  It is easy to work around this objection.  The deterministic relationship disappears when an arbitrarily small
  stochastic noise is added to the joint distribution.  To be precise, let $X_{1},X_{2}$ be independent binary random
  variables, and let $X_{3}$ be binary with
  \begin{equation*}
    P(X_{3}=x_{3}|X_{1}=x_{1},X_{2}=x_{2}) =
    \begin{cases}
      1 - \epsilon, & \text{ if }x_{3} = x_{1}\oplus x_{2}, \\
      \epsilon, & \text{ otherwise,}
    \end{cases}
  \end{equation*}
  for $0\le\epsilon\le1$.  For $\epsilon=0$, the example from the proof is recovered.
  Assuming that the partial information terms depend continuously on this joint distribution, the partial information
  term $I_{\partial}\big(S;(X_{1}X_{2}),(X_{1}X_{3}),(X_{2}X_{3})\big)$ will still be negative for small~$\epsilon>0$.
  Thus, assuming continuity, the conclusion of Theorem~\ref{thm:incompatibility} still holds true when the information
  decomposition according to the full partial information lattice is only considered for random variables that do not
  satisfy any functional deterministic constraint.
\end{Remark}
\begin{Remark}
  Analyzing the proof of Theorem~\ref{thm:incompatibility}, one sees that the independent identity axiom
  (Lemma~\ref{lem:weak-identity}) is the main ingredient to arrive at the contradiction.  The same property also arises
  in the other uses of the XOR example \citep{BROJ13:Shared_information,RBOJ14:Reconsidering_unique_information}.
\end{Remark}

%%%%%%%%%%%%%%%%%%%%%%%%%%%%%%%%%%%%%%%%%%
\section{Discussion}
\label{sec:discussion}

Perfect secret sharing schemes correspond to systems of random variables in which it is very clearly
specified ``who knows what.''  In such a system, it is easy to assign intuitive values to the shared information nodes
in the partial information lattice, and one may conjecture that the intuition behind this assignment is the same
intuition that underlies the Williams and Beer axioms, which define the partial information lattice.  Moreover,
following the same intuition, independent combinations of perfect secret sharing schemes can be used as a tool to
construct systems of random variables with prescribable (nonnegative) values of partial information.

Unfortunately, this extension to independent combinations of perfect secret sharing schemes is not without problems: By
Theorem~\ref{thm:incompatibility}, it leads to decompositions with negative partial information terms.  But what does it
mean that the examples derived from the same intuition as the Williams and Beer axioms contradict the same axioms in
this way?  Is this an indication that the whole idea of information decomposition does not work (and that the question
posed in the first paragraph of the introduction cannot be answered affirmatively)?

There are several ways out of this dilemma.
The first solution is to assign different values to combinations of perfect secret sharing schemes.  This solution will
not be pursued further in this text, as it would change the interpretation of the information decomposition as measuring
``who knows what.''

The second solution is to accept negative partial values in the information decomposition.  It has been argued that
negative values of information can be given an intuitive interpretation in terms of confusing or misleading information.
For event-wise (also called ``local'') information quantities, such as the event-wise mutual information
$i(s;x) = \log(p(s)/p(s|x))$, this interpretation goes back to the early days of information
theory~\cite{Fano61:Transmission_of_information}.  Sometimes, this phenomenon is % wrongly
called ``misinformation'' \citep{Ince17:Iccs,WibralLizierPriesemann14:Bits_from_Biology}.  However, in the usual
language, misinformation refers to ``false or incorrect information, especially when it is intended to trick
someone''~\cite{Macmillan:Dictionary}, which is not the effect that is modelled here.  Thus, the word misinformation
should be avoided, in order not to mislead the reader into the wrong intuition.

While negative event-wise information quantities are well-understood, the situation is more problematic for average
quantities.  When an agent receives side-information in the form of the value $x$ of a relevant random variable~$X$, she
changes her strategy.  While the prior strategy should be based on the prior distribution~$p(S)$, the new strategy
should be based on the posterior $p(S|X=x)$.  Clearly, in a probabilistic setting, any change of strategy can lead to a
better or worse result in a single instance.  On average, though, side-information never hurts (and it is never
advantageous on average to ignore side-information), which is why the mutual information is never negative.  Similarly,
it is natural to expect non-negativity of other information quantities.  It is difficult to imagine how correct
side-information (or an aspect thereof) can be misleading on average.  The situation is different for incorrect
information, where the interpretation of a negative value is much easier.

More conceptually, I would suspect that an (averaged) information quantity that may change its sign actually conflates different
aspects of information\footnote{One can argue whether the same should be true for event-wise quantities.  Recently, \cite{Ince17:PED} suggested to also write the event-wise mutual information as a difference of non-negative quantities.}, just as the interaction information (or co-information) conflates synergy and
redundancy~\cite{WilliamsBeer:Nonneg_Decomposition_of_Multiinformation}.

In any case, allowing negative partial values alters the interpretation of an information decomposition to a point where
it is questionable whether the word ``decomposition'' is still appropriate.
When decomposing an object into parts, the parts should in some reasonable way be sub-objects.  For example, in a
Fourier decomposition of a function, the Fourier components are never larger than the function (in the sense of the
$L^{2}$-norm), and the sum of the squared $L^{2}$-norms of the Fourier coefficients equals the squared $L^{2}$-norm of
the original function.
As another example, given a (positive) amount of money and two investment options, it may indeed be possible to invest a
negative share of the total amount into one of the two options in order to increase the funds that can be invested in
the second option.  However, such short selling is regulated in many countries with much stronger rules than ordinary
trading.

I do not claim that an information decomposition with negative partial information terms cannot possibly make sense.
However, it has to be made clear precisely how to interpret negative terms, and it is important to distinguish between
correct information that leads to a suboptimal decision due to unlikely events happening (``bad luck'') and incorrect
information that leads to decisions being based on the wrong posterior probabilities (as opposed to the ``correct''
conditional probabilities).

A third solution is to change the underlying lattice structure of the decomposition.  A first step in this direction was
done by~\citet{ChicharroPanzeri17:Dual_Decompositions} who propose to decompose mutual information according to subsets
of the partial information lattice.  However, it is also conceivable that the lattice has to be enlarged.

Williams and Beer derived the partial information lattice from their axioms together with the assumption that everything
can be expressed in terms of shared information (that is, according to ``who knows what'').  Shared information is
sometimes equivalently called \emph{redundant information}, but it may be necessary to distinguish the two.  Information
that is shared by several random variables is information that is accessible to each single random variable, but
redundancy can also arise at higher orders.  An example is the infamous XOR example from the proof of
Theorem~\ref{thm:incompatibility}: In this example, each pair $X_{i},X_{j}$ is independent and contains of two bits, but
the total system $X_{1},X_{2},X_{3}$ has only two bits.  Therefore, there is one bit of redundancy.  However, this
redundancy bit is not located anywhere specifically: It is not contained in either of $X_{1},X_{2},X_{3}$, and thus it
is not shared information.  Since the redundant bit is not part of~$X_{1}$, it is not ``shared'' by $X_{1}$ in this
sense.  This phenomenon corresponds to the fact that random variables can be pairwise independent without being
independent.

This kind of higher order redundancy does not have a place in the partial information lattice, so it may be that nodes
corresponding to higher order redundancy have to be added.  When the lattice is enlarged in this way, the structure of
the Möbius inversion is changed, and it is possible that the resulting lattice leads to nonnegative partial information
terms, without changing those cumulative information values that are already present in the original lattice.
If this approach succeeds, the answer to the question from the introduction will be negative: Simply classifying
information according to ``who knows what'' (i.e. shared information) does not work, since it does not capture higher
order redundancy.  The analysis of extensions of the partial information lattice is scope for future work.

% %%%%%%%%%%%%%%%%%%%%%%%%%%%%%%%%%%%%%%%%%%
% \section{Conclusions}

% This section is not mandatory, but can be added to the manuscript if the discussion is unusually long or complex.

%%%%%%%%%%%%%%%%%%%%%%%%%%%%%%%%%%%%%%%%%%
\vspace{6pt} 

%%%%%%%%%%%%%%%%%%%%%%%%%%%%%%%%%%%%%%%%%%
%% optional
% \supplementary{The following are available online at www.mdpi.com/link, Figure S1: title, Table S1: title, Video S1: title.}

%%%%%%%%%%%%%%%%%%%%%%%%%%%%%%%%%%%%%%%%%%
\newcommand{\acknowledgmentstext}{I thank Fero Matúš for teaching me about secret sharing schemes.  I am grateful to Guido
  Montúfar and Pradeep Kr.~Banerjee for their remarks about the manuscript, and to Nils Bertschinger, Jürgen
  Jost and Eckehard Olbrich for many inspiring discussions on the topic.  I thank the reviewers for many comments, in
  particular concerning the discussion.
  I thank the organizers and participants of the PID workshop in December 2016 in Frankfurt, where the material was first presented.
}
\ifEntropy
\acknowledgments{\acknowledgmentstext}
\else
\paragraph{\textbf{Acknowledgments}}
\acknowledgmentstext
\fi

%%%%%%%%%%%%%%%%%%%%%%%%%%%%%%%%%%%%%%%%%%
\ifEntropy
\conflictsofinterest{The author declares no conflict of interest.} 
\fi

% %%%%%%%%%%%%%%%%%%%%%%%%%%%%%%%%%%%%%%%%%%
% %% optional
\ifEntropy
\appendixtitles{yes}
\appendixsections{one}
\else
\vfill
\fi
\appendix
\section{Combined secret sharing properties for small~\texorpdfstring{$k$}{k}}

This section discusses the defining equation~\eqref{eq:comb-ssp} of the combined secret sharing property for $k=1$ and
$k=2$.  The case $k=1$ is incorporated in the definition of a combination of perfect secret sharing schemes: The
following lemma implies that any measure of shared information that satisfies self-redundancy
satisfies~\eqref{eq:comb-ssp} for $k=1$.  Recall that Williams and Beer's self-redundancy axiom implies that
$I_{\cap}(S;X_{A}) = I(S;X_{A})$.
\begin{Lemma}
  \label{lem:k=1}
  Let $(S_{1},\dots,S_{l},X_{1},\dots,X_{n})$ be a combination of perfect secret sharing schemes with access structures
  $\Acal_{1},\dots,\Acal_{l}$.  Then
  \begin{equation*}
    I\big((S_{1},\dots,S_{l});X_{A}\big) = H\big(\{ S_{i} : A\in\Acal_{i}\}\big).
  \end{equation*}
\end{Lemma}
\begin{proof}
  Suppose that the secret for which $A$ is authorized are $S_{1},\dots,S_{m}$.  Then
  \begin{multline*}
    H(S_{1},\dots,S_{l}|X_{A}) = H(S_{1},\dots,S_{m}|X_{A}) + H(S_{m+1},\dots,S_{l}|S_{1},\dots,S_{m},X_{A}) \\
     = H(S_{m+1},\dots,S_{l}|S_{1},\dots,S_{m},X_{A})
    \le H(S_{m+1},\dots,S_{l}) \le \sum_{i=m+1}^{l}H(S_{i}).
  \end{multline*}
  On the other hand,
  \begin{multline*}
    H(S_{m+1},\dots,S_{l}|S_{1},\dots,S_{m},X_{A})
    = \sum_{i=m+1}^{l} H(S_{i}|S_{1},\dots,S_{i-1},X_{A}) \\
    \ge \sum_{i=m+1}^{l} H(S_{i}|S_{1},\dots,S_{i-1},S_{i+1},\dots,S_{l},X_{A})
    = \sum_{i=m+1}^{l} H(S_{i}).
  \end{multline*}
  By independence (remark after Definition~\ref{def:sss_combi}), $\sum_{i=m+1}^{l}H(S_{i}) = H(S_{m+1},\dots,S_{l})$ and
  $\sum_{i=1}^{m}H(S_{i}) = H(S_{1},\dots,S_{m})$.  Thus,
  \begin{equation*}
    I\big((S_{1},\dots,S_{l});X_{A}\big) =
    H(S_{1},\dots,S_{l}) - H(S_{1},\dots,S_{l}|X_{A}) = H(S_{1},\dots,S_{m}).
    \qedhere
  \end{equation*}
\end{proof}

The next result shows that the bivariate measure of shared information $\widetilde{SI}(S;X,Y)$ proposed by
\citet{BROJA13:Quantifying_unique_information} satisfies Eq.~\eqref{eq:comb-ssp} for~$k\le 2$.  The reader is referred
to \emph{loc. cit.} for definitions and elementary properties of~$\widetilde{SI}$.
\begin{Proposition}
  \label{lem:k=2}
  Let $(S_{1},\dots,S_{l},X_{1},\dots,X_{n})$ be a combination of perfect secret sharing schemes with access structures
  $\Acal_{1},\dots,\Acal_{l}$,
%  The combined secret sharing property implies
  Then
  \begin{equation*}
    \widetilde{SI}\big((S_{1},\dots,S_{l});X_{A_{1}},X_{A_{2}}\big) = H\big(\{ S_{i} : A\in\Acal_{1}\cap\Acal_{2}\}\big).
  \end{equation*}
\end{Proposition}
\begin{proof}
  For given $A_{1},A_{2}$, suppose that $S_{1},\dots,S_{m}$ are the secrets for which at least one of $A_{1}$ or $A_{2}$
  is authorized and that $S_{m+1},\dots,S_{l}$ are the secrets for which neither $A_{1}$ nor $A_{2}$ is authorized
  alone.

  Let $P$ be the joint distribution of $S_{1},\dots,S_{l},X_{A_{1}},X_{A_{2}}$.  Let $\Delta_{P}$ be the set of
  alternative joint distributions for $S_{1},\dots,S_{l},X_{A_{1}},X_{A_{2}}$ that have the same marginal distributions
  as $P$ on the subsets $(S_{1},\dots,S_{l},X_{A_{1}})$ and~$(S_{1},\dots,S_{l},X_{A_{2}})$.  According to the
  definition of $\widetilde{SI}$, we need to compare $P$ with the elements of~$\Delta_{P}$ and find the maximum of
  $H_{Q}\big((S_{1},\dots,S_{l})\big|X_{A_{1}},X_{A_{2}}\big)$ over~$Q\in\Delta_{P}$, where the subscript to $H$ indicates
  with respect to which of these joint distributions the conditional entropy is evaluated.

  Define a distribution $Q^{*}$ for $S_{1},\dots,S_{l},X_{A_{1}},X_{A_{2}}$ by
  \begin{equation*}
    Q^{*}(s_{1},\dots,s_{l},x_{1},x_{2})
    = P(s_{1},\dots,s_{l}) P(x_{A_{1}} = x_{1}|s_{1},\dots,s_{l}) P(x_{A_{2}} = x_{2}|s_{1},\dots,s_{l}).
  \end{equation*}
%  (this construction corresponds to the distribution $Q^{*}$ by \citet{BROJA13:Quantifying_unique_information}).
  Then $Q^{*}\in\Delta_{P}$.  Under $P$, the secrets $S_{m+1},\dots,S_{l}$ are independent of $X_{A_{1}}$ (marginally)
  and independent of~$X_{A_{2}}$, and so $S_{m+1},\dots,S_{l}$ are independent of the pair $(X_{A_{1}},X_{A_{2}})$
  under~$Q^{*}$.  On the other hand, $S_{1},\dots,S_{m}$ are a function of either $X_{A_{1}}$ or $X_{A_{2}}$ under~$P$,
  and so $S_{1},\dots,S_{m}$ is a function of $(X_{A_{1}},X_{A_{2}})$ under~$Q^{*}$.  Thus,
  \begin{equation*}
    H_{Q^{*}}(S_{1},\dots,S_{l}|X_{A_{1}},X_{A_{2}})
    = H_{Q^{*}}(S_{m+1},\dots,S_{l}) = H_{P}(S_{m+1},\dots,S_{l}).
  \end{equation*}
  On the other hand, under any joint distribution $Q\in\Delta_{P}$, the secrets $S_{1},\dots,S_{m}$ are functions of
  $X_{A_{1}},X_{A_{2}}$, whence
  \begin{equation*}
    H_{Q}(S_{1},\dots,S_{l}|X_{A_{1}},X_{A_{2}})
    \le H_{Q}(S_{m+1},\dots,S_{l}) = H_{P}(S_{m+1},\dots,S_{l}).
  \end{equation*}
  It follows that $Q^{*}$ solves the optimization problem in the definition of~$\widetilde{SI}$.

  Suppose that the secrets for which $X_{A_{1}}$ is authorized are $S_{1},\dots,S_{r}$ and that the secrets for which
  $X_{A_{2}}$ is authorized are $S_{s},\dots,S_{m}$ (with $1\le r,s\le m$).  One computes
  \begin{align*}
    I_{Q^{*}}\big((S_{1},\dots,S_{l});X_{A_{1}}\big|X_{A_{2}}\big)
    &= H(S_{1},\dots,S_{s-1}) = \sum_{i=1}^{s-1}H(S_{i})
    \quad\text{ and }\\
    I_{Q^{*}}\big((S_{1},\dots,S_{l});X_{A_{1}}\big)
    &= H(S_{1},\dots,S_{r}) = \sum_{i=1}^{r}H(S_{i}),
  \end{align*}
  whence
  \begin{multline*}
    \widetilde{SI}\big((S_{1},\dots,S_{l});X_{A_{1}},X_{A_{2}}\big)
    \ifEntropy\else\\\fi
    = I_{Q^{*}}\big((S_{1},\dots,S_{l});X_{A_{1}}\big) - I_{Q^{*}}\big((S_{1},\dots,S_{l});X_{A_{1}}\big|X_{A_{2}}\big) \\
    = \sum_{i=s}^{r}H(S_{i}) = H(S_{s},\dots,S_{r}).
    \qedhere
  \end{multline*}
%  as required by Eq.~\eqref{eq:comb-ssp}.
\end{proof}

%=====================================
% References, variant B: external bibliography
%=====================================
\ifEntropy
\externalbibliography{yes}
\else
\bibliographystyle{mdpi}
\fi
\bibliography{../Info}

\begin{thebibliography}{-------}
\providecommand{\natexlab}[1]{#1}

\bibitem[Williams and
  Beer(2010)]{WilliamsBeer:Nonneg_Decomposition_of_Multiinformation}
Williams, P.; Beer, R.
\newblock Nonnegative Decomposition of Multivariate Information.
\newblock {\em arXiv:1004.2515v1} {\bf 2010}.

\bibitem[Beimel(2011)]{Beimel:Secret_sharing_survey}
Beimel, A.
\newblock Secret-sharing Schemes: A Survey.
\newblock  Proceedings of the Third International Conference on Coding and
  Cryptology; Springer-Verlag: Berlin, Heidelberg,  2011; pp. 11--46.

\bibitem[Maurer and Wolf(1997)]{MaurerWolf97:intrinsic_conditional_MI}
Maurer, U.; Wolf, S.
\newblock The intrinsic conditional mutual information and perfect secrecy.
\newblock  Proc. IEEE ISIT,  1997.

\bibitem[Csiszar and
  Narayan(2004)]{CsiszarNarayan04:SEcrecy_capacities_for_multiple_terminals}
Csiszar, I.; Narayan, P.
\newblock Secrecy capacities for multiple terminals.
\newblock {\em IEEE Transactions on Information Theory} {\bf 2004}, {\em
  50},~3047--3061.

\bibitem[Ito \em{et~al.}(1987)Ito, Saito, and
  Nishizeki]{ItoSaitoNishizeki87:General_secret_sharing_schemes}
Ito, M.; Saito, A.; Nishizeki, T.
\newblock Secret sharing scheme realizing general access structure.
\newblock  Proceedings of the IEEE Global Telecommunication Conf.,  1987, pp.
  99--102.

\bibitem[Bertschinger \em{et~al.}(2013)Bertschinger, Rauh, Olbrich, and
  Jost]{BROJ13:Shared_information}
Bertschinger, N.; Rauh, J.; Olbrich, E.; Jost, J.
\newblock Shared Information --- New Insights and Problems in Decomposing
  Information in Complex Systems. In {\em Proc. ECCS 2012}; Springer,  2013;
  pp. 251--269.

\bibitem[Rauh \em{et~al.}(2014)Rauh, Bertschinger, Olbrich, and
  Jost]{RBOJ14:Reconsidering_unique_information}
Rauh, J.; Bertschinger, N.; Olbrich, E.; Jost, J.
\newblock Reconsidering unique information: Towards a multivariate information
  decomposition.
\newblock  Proc. IEEE ISIT,  2014, pp. 2232--2236.

\bibitem[Barrett(2014)]{Barrett2014:Gaussian_information_decomposition}
Barrett, A.B.
\newblock An exploration of synergistic and redundant information sharing in
  static and dynamical Gaussian systems.
\newblock {\em CoRR} {\bf 2014}, {\em abs/1411.2832}.

\bibitem[Harder \em{et~al.}(2013)Harder, Salge, and
  Polani]{HarderSalgePolani2013:Bivariate_redundancy}
Harder, M.; Salge, C.; Polani, D.
\newblock A Bivariate measure of redundant information.
\newblock {\em Phys. Rev. E} {\bf 2013}, {\em 87},~012130.

\bibitem[Bertschinger \em{et~al.}(2014)Bertschinger, Rauh, Olbrich, Jost, and
  Ay]{BROJA13:Quantifying_unique_information}
Bertschinger, N.; Rauh, J.; Olbrich, E.; Jost, J.; Ay, N.
\newblock Quantifying unique information.
\newblock {\em Entropy} {\bf 2014}, {\em 16},~2161--2183.

\bibitem[Ince(2017)]{Ince17:Iccs}
Ince, R.
\newblock Measuring multivariate redundant information with pointwise common
  change in surprisal.
\newblock {\em Entropy} {\bf 2017}, {\em 19},~318.

\bibitem[Chicharro and Panzeri(2017)]{ChicharroPanzeri17:Dual_Decompositions}
Chicharro, D.; Panzeri, S.
\newblock Synergy and Redundancy in Dual Decompositions of Mutual Information
  Gain and Information Loss.
\newblock {\em Entropy} {\bf 2017}, {\em 19}.

\bibitem[Fano(1961)]{Fano61:Transmission_of_information}
Fano, R.M.
\newblock {\em Transmission of Information}; MIT Press: Cambridge, MA,  1961.

\bibitem[Wibral \em{et~al.}(2015)Wibral, Lizier, and
  Priesemann]{WibralLizierPriesemann14:Bits_from_Biology}
Wibral, M.; Lizier, J.T.; Priesemann, V.
\newblock Bits from Brains for Biologically Inspired Computing.
\newblock {\em Frontiers in Robotics and AI} {\bf 2015}, {\em 2},~5.

\bibitem[{Macmillan Publishers Limited}(retrieved on
  2017/10/05)]{Macmillan:Dictionary}
{Macmillan Publishers Limited}.
\newblock Macmillan Dictionary.
\newblock Available at \url{http://www.macmillandictionary.com/},  retrieved on
  2017/10/05.

\bibitem[Ince(2017)]{Ince17:PED}
Ince, R.
\newblock The Partial Entropy Decomposition: Decomposing multivariate entropy
  and mutual information via pointwise common surprisal.
\newblock {\em arXiv:1702.01591} {\bf 2017}.

\end{thebibliography}

\end{document}